\documentclass[a4paper,UKenglish,cleveref, autoref,final]{lipics-v2019}

\usepackage{xcolor}
\usepackage{fixme}
\FXRegisterAuthor{sb}{asb}{SB}
\FXRegisterAuthor{kg}{akg}{KG}
\fxsetup{
  status=draft,
  layout=inline,
  theme=color
}
\colorlet{fxnote}{blue}
\colorlet{fxnotebg}{yellow!75}

\usepackage{paralist}
\usepackage{enumitem}
\usepackage{nicefrac}

\newcommand{\OPT}{\operatorname{\text{\textsc{opt}}}}
\newcommand{\Vol}{\operatorname{\text{\textsc{vol}}}}
\newcommand{\Val}{\operatorname{\text{\textsc{costs}}}}
\newcommand{\Sol}{\operatorname{\text{\textsc{sol}}}}

\newcommand{\costs}{\operatorname{\text{\textsc{costs}}}}
\newcommand{\ITEMS}{\operatorname{\text{ITEMS}}}
\newcommand{\alg}[1]{\mathsf{#1}}

\newcommand{\shortrnote}[1]{ &  & \text{\footnotesize\llap{// #1}}}

\theoremstyle{plain}
\newtheorem{assumption}{Assumption}

\title{Robust Online Algorithms for Dynamic Problems} 

\author{\ }{\ }{}{}{}
\authorrunning{\ }
\Copyright{\ }

\author{Sebastian Berndt}{Department of Computer Science, Kiel University,
  Germany}{seb@informatik.uni-kiel.de}{}{Supported
 by DFG Project, "Robuste Online-Algorithmen für Scheduling- und Packungsprobleme", JA 612 /19-1}
\author{Valentin Dreismann}{University of Warwick, United Kingdom}{valentin.dreismann@posteo.net}{}{}
\author{Kilian Grage}{Department of Computer Science, Kiel University,
  Germany}{kig@informatik.uni-kiel.de}{}{Supported by by GIF-Project "Polynomial Migration for Online Scheduling".}
\author{Klaus Jansen}{Department of Computer Science, Kiel University,
  Germany}{kj@informatik.uni-kiel.de}{}{}
\author{Ingmar Knof}{Kiel University,
  Germany}{ingmar.knof@posteo.de}{}{}

\authorrunning{S.~Berndt and V.~Dreismann and K.~Grage and K.~Jansen and I.~Knof}

\Copyright{Sebastian Berndt and Valentin Dreismann and Kilian Grage and Klaus
  Jansen and Ingmar Knof}



\nolinenumbers 

\hideLIPIcs  

\EventEditors{John Q. Open and Joan R. Access}
\EventNoEds{2}
\EventLongTitle{42nd Conference on Very Important Topics (CVIT 2016)}
\EventShortTitle{CVIT 2016}
\EventAcronym{CVIT}
\EventYear{2016}
\EventDate{December 24--27, 2016}
\EventLocation{Little Whinging, United Kingdom}
\EventLogo{}
\SeriesVolume{42}
\ArticleNo{23}

\begin{document}

\maketitle

\begin{abstract}
  Online algorithms that allow a small amount of migration or recourse have been
  intensively studied in the last years. 
  They are essential in the design of competitive algorithms for dynamic
  problems, where objects can also depart from the instance. 
  In this work, we give a general framework to obtain so called robust online
  algorithms for these dynamic problems: 
  these online algorithm achieve an asymptotic competitive ratio of
  $\gamma+\epsilon$ with migration $O(1/\epsilon)$, where $\gamma$ is the best
  known offline asymptotic approximation ratio.
  In order to use our framework, one only needs to construct a suitable online
  algorithm for the \emph{static} online case, where items never depart.
  To show the usefulness of our approach, we improve upon the best known robust
  algorithms for the dynamic versions of generalizations of Strip Packing and
  Bin Packing, including the first robust algorithm for general $d$-dimensional
  Bin Packing and Vector Packing. 
\end{abstract}

\section{Introduction}
Online algorithms are a very natural way to deal with uncertain inputs and the
need for a sequence of good solutions throughout the evolution of an instance.
For a surprisingly large number of problems, one can obtain online algorithms
that produce solutions within a constant ratio of the optimal solutions. 
The worst-case ratio between an optimal solution and a solution produced by the
online algorithm is called the \emph{competitiveness} of the algorithm.
For most algorithms with constant competitiveness, the algorithms rely heavily
on the fact that the instances involve in a \emph{monotone} way: 
every object that becomes part of the instance will stay part of the instance
forever. 
This monotonicity property is often not present in real-world applications,
where the objects might be removed from the instance later on. 
These objects might be used at a different place, they might expire, or they are
no longer relevant. 
Hence, in order to still give a performance guarantee, the online algorithms
need to be able to modify parts of an existing solutions. 
Clearly, such a modification might be costly and should thus be minimized.
If no such boundary on the modification is given, one could easily use an
offline approximation algorithm. 
A natural way to measure the amount of modification needed is the
\emph{migration factor}: 
it compares the total size of modified objects with the size of the newly
inserted or removed object.
An algorithm with a bounded migration factor roughly translates to the fact that
the insertion or departure of a small object can only lead to small changes in
the structure of the current solution. 
On the other hand, if a large (and thus impactful) object is inserted or
removed, we are allowed to modify a larger part of the solution. 

It is easy to see that online algorithms can in general not achieve the same
solution quality as offline algorithms. 
The question how much these two values differ is one of the central questions in
the field of online and approximation algorithms. 
Clearly, an online algorithm that is allowed a certain amount of migration is
able to gap between these extremes.
If the amount of migration needed directly corresponds to the improvement of the
solution guarantee, we call such an algorithm \emph{robust}.
As an example, consider the well-studied Bin Packing problem.
No online algorithm (without migration) can achieve a competitiveness smaller
than $3/2$~\cite{yao1980new}, while an offline algorithm with asymptotic
approximation ratio of $1+\epsilon$ can be obtained in polynomial
time~\cite{karmarkar1982efficient}.
A robust online algorithm would then have competitiveness $1+\epsilon$ and
migration factor $f(1/\epsilon)$ for some function.
In other words, we only need to increase the migration if we want to obtain a
better solution.
Such robust online algorithms thus have a continuous behavior between the
performance of the best offline algorithm ($\epsilon \to 0$) and the performance
of the best online algorithm ($\epsilon \to \infty$). 

In this paper, we describe a general framework to construct robust online
algorithms for dynamic online problems. 
The only ingredients needed for our framework are (i) a good offline
approximation algorithm and a (ii) suitably designed online algorithm for the
\emph{static} case. 
Our framework is applicable to many problems, especially in the field of
geometric packing problems.
To show the versatility of our approach, we use it to improve upon existing
robust algorithms and to construct new robust algorithms.

\subsection{Related Works}
The migration factor model was introduced by Sanders, Sivadasan and
Skutella~\cite{DBLP:journals/mor/SandersSS09}. They studied the classical
problem of minimizing the makespan on parallel identical machines and made use
of sensitivity results of integer programming to obtain a robust
$1+\epsilon$-competitive algorithm. This
work spanned a lot of follow up works for a wide range of different problems:
\begin{itemize}
\item Skutella and Verschae~\cite{DBLP:journals/mor/SkutellaV16} also studied
  the same problem and were able to obtain a robust $1+\epsilon$-competitive
  algorithm for the dynamic case, where jobs may depart. 
  They also studied the dual version of makespan minimization problem, where the
  minimum load should be maximized. 
  This problem (called Machine Covering or Santa Claus) also admits a robust
  $1+\epsilon$-competitive algorithm even in the dynamic case. 
  Skutella and Verschae also proved that for all $\epsilon > 0$ there is no
  online algorithm for the Machine Covering problem that achieves competitive
  ratio $20/19-\epsilon$ with worst-case migration $f(1/\epsilon)$ for any
  function $f$. 
  G{\'{a}}lvez \emph{et al.}~\cite{DBLP:conf/esa/GalvezSV18} gave two simple
  competitive robust algorithms (one that is $1{.}7+\epsilon$-competitive and
  one that is $4/3+\epsilon$-competitive) with polynomial migration factor for
  the static version of the problem of makespan minimization. 
\item For the Bin Packing problem, Epstein and
  Levin~\cite{DBLP:journals/mp/EpsteinL09} gave the first robust
  $1+\epsilon$-competitive algorithm for Bin Packing based on the same
  sensitivity results for integer programming.
  Jansen and Klein~\cite{DBLP:conf/icalp/JansenK13} designed new techniques in
  order to obtain a migration factor polynomial in $1/\epsilon$. 
  These techniques were improved by Berndt \emph{et
    al.}~\cite{DBLP:conf/approx/BerndtJK15} to also handle the dynamic version
  of the Bin Packing problem where items can also depart. 
  In all of these works, the migration factor is worst-case and can thus not be
  saved up for later use. 
  In contrast, Feldkord \emph{et al.}~\cite{DBLP:conf/icalp/FeldkordFGGKRW18}
  presented a simple robust $1+\epsilon$-competitive algorithm with amortized
  migration factor that also works for the dynamic case. 
  Berndt \emph{et al.}~also showed that a worst-case migration factor of
  $\Omega(1/\epsilon)$ is needed~\cite{DBLP:conf/approx/BerndtJK15}. 
  This lower bound was shown to also hold for amortized migration by Feldkord
  \emph{et al.}~\cite{DBLP:conf/icalp/FeldkordFGGKRW18}. 
  Feldkord \emph{et al.}~also studied a problem variant, where the migration
  needed for an item does not correspond to its size. 
\item For the makespan problem, where jobs can be scheduled preemptively
  (i.\,e.~they can be split, but parts are not allowed to run simultaneously),
  Epstein and Levin~\cite{DBLP:journals/algorithmica/EpsteinL14} designed an
  optimal online algorithm with migration factor $1-1/m$, where $m$ is the
  number of machines. 
  Again, the migration factor used here is worst-case. 
  They also showed that exact algorithms for the makespan minimization problem
  on uniform machines and for identical machines in the restricted assignment
  have worst-case migration factor at least $\Omega(m)$. 
\end{itemize}

The offline variants of the geometric packing problems have also been studied
intensively. See e.\,g.~the survey of Christensen \emph{et
  al.}~\cite{christensen2017survey} and the references therein. 

\subsection{Our Results}
We present a very general framework that allows to construct robust online
algorithms that have optimal amortized migration factor of $O(1/\epsilon)$ and
achieve the same competitive ratio as the best known offline approximation
algorithm (up to an additive error of $\epsilon$). 
This framework is suitable for many different optimization problems, especially
geometric packing problems. 
The algorithms created by our framework can all deal with the dynamic versions
of these problems, where items can also depart from the instance. 
To present a general framework, we introduce the notion of \emph{flexible online
  algorithms} and show that such algorithms can be combined with offline
approximation algorithms under certain circumstances.
We then show the versatility of our approach by looking at several well-studied
problems including generalizations of Strip Packing and Bin Packing. 
We give robust algorithms for the multi-dimensional variants of these problems,
where we can also handle both the departure and the rotation of items. 
This improves and generalizes several known results and gives the first robust
algorithms for all of the other problems. 
Additionally, our compact and clean framework gives much easier algorithms
compared with the previously known algorithms.

\section{Online Algorithms for Dynamic Problems: A Framework}

\subsection{Preliminaries}
While the techniques presented in this paper also work for maximization
problems, we will focus on minimization problems to improve the accessibility of
our results.

\begin{definition}
  Let $\ITEMS\subseteq \{0,1\}^{*}$ be a prefix-free set describing some items. A \emph{minimization problem} $\Pi=(\mathcal{I},\Sol,\Val)$ consist of a set
  of \emph{instances} $\mathcal{I}\subseteq \ITEMS$, a mapping $\Sol$ that maps an
  instance $I\in \mathcal{I}$ to a non-empty set of \emph{feasible solution}
  $\Sol(I)$, and a mapping $\Val$ that maps a solution $S\in \Sol(I)$ to its \emph{costs}
  $\Val(S)\in \mathbb{Q}_{\geq 0}$. A solution $S^{*}\in \Sol(I)$ is called
  \emph{optimal}, if $\Val(S^{*})\leq \Val(S)$ for all $S\in \Sol(I)$. We denote
  the cost of any optimal solution $S^{*}$ by $\OPT(I):= \Val(S^{*})$. 
\end{definition}
For the sake of simplicity, we also sometimes treat $\Pi$ and $\mathcal{I}$
interchangeably and also write $I\in \Pi$ to denote that $I$ is an instance. 

Throughout this paper, we make the following two natural assumptions that hold
for all problems considered in this work:
\begin{assumption}\
  
  \label{assumption}
  \begin{compactenum}
  \item For every instance $I$, every instance $I'\subseteq I$,
and every solution $S\in \Sol(I)$, the solution $S\restriction_{I'}$ induced by
the items in $I'$ is also feasible, i.\,e.~$S\restriction_{I'}\in \Sol(I')$.
Hence, removing some items from a feasible solution results in a feasible
solution of the remaining items.
\item For every instance $I$, and every instance $I'\subseteq I$, we have
  $\OPT(I')\leq \OPT(I)$. Hence, removing items from an instance can only
  decrease the optimum. 
  \end{compactenum}
\end{assumption}

In the classical \emph{offline version} of $\Pi=(\mathcal{I}, \Sol,\Val)$, we
are given the complete instance $I\in \mathcal{I}$ all in once. 
An \emph{$\alpha$-approximation} for $\Pi$ is a polynomial-time algorithm $\alg{A}$ such
that for all instances $I\in \mathcal{I}$, we have $\alg{A}(I)\leq \alpha
\OPT(I)+c$ for some constant $c$ not depending on the instance. 
Here, $\alg{A}(I)$ is the value of the solution produced by $\alg{A}$.

In the \emph{online version} of a minimization
problem $\Pi=(\mathcal{I}, \Sol, \Val)$, an \emph{online instance} $I$ is a sequence of
instances $I_{1},I_{2},\ldots,I_{|I|}\in \mathcal{I}$ such that $|I_{t}\triangle
I_{t+1}|=1$ for all $t=1,\ldots,|I|-1$, where $\triangle$ is the symmetric
difference of two sets. This means that we \emph{insert} a new item $i^{*}$ (if
$I_{t+1}\setminus I_{t}=\{i^{*}\}$) or an item $i^{*}$ \emph{departs} (if
$I_{t}\setminus I_{t+1}=\{i^{*}\}$). The set of online instances of $\Pi$ is
denoted as $\mathcal{I}^{\text{on}}$. 
An \emph{online algorithm} $\alg{A}$
maintains a sequence of solutions $S_{1},\ldots,S_{t}$, where $S_{t}\in
\Sol(I_{t})$ and furthermore, $S_{t+1}(i)=S_{t}(i)$ for all $i\in I_{t+1}\cap
I_{t}$.  In the \emph{static online
  version}, we only have insertions and thus $I_{t}\subsetneq I_{t+1}$. A non-static problem is called \emph{dynamic}. 
We say that an online algorithm $\alg{A}$ producing such a sequence of solutions $S_{1},S_{2},\ldots$ is
 \emph{$\beta$-competitive}, if
\begin{align*}
  \Val(S_{t})\leq \beta\cdot \OPT(I_{t})+c
\end{align*}
for all $t=1,\ldots,|I|$ and some constant $c$ not depending on the instance.
This notion of competitiveness is sometimes also called \emph{asymptotic
  competitiveness} in contrast to the notion of \emph{absolute competitiveness},
where no additional additive term $c$ is allowed. 
Whenever we talk about competitiveness, this is with regard to the notion of
asymptotic competitiveness. 

\paragraph*{Migration}
Achieving bounded competitiveness is usually impossible for dynamic problems,
even for very simple problems such as Bin Packing where every item has the same
size. 
This is simply due to the fact that an online algorithm is not allowed to change
the placement of already placed items. 
This very strict restriction thus comes with a high cost with regard to the
competitiveness. 
As many real-world applications are not only static, but allow the departure of
items, one must thus me more flexible. 
We will thus allow a small amount of repacking to be able to handle dynamic
problems. 
The model of repacking we use is called the \emph{amortized migration factor
  model}. 
In this setting, every item $i\in \ITEMS$ comes with a \emph{size} $v_{i}\in
\mathbb{Q}_{\geq 0}$.
Usually this size will be the space needed for an item, like it's total area or
volume or some identical criteria like the side length of a hypercube.
For a set of items $I\subseteq \ITEMS$, we denote by $\Vol(I)=\sum_{i\in
  I}v_{i}$ the complete volume of $I$. 
If $I\in \mathcal{I}^{\mathrm{on}}$ is an online instance and
$\vec{S}=(S_{1},\ldots,S_{|I|})$ is a sequence of solutions with $S_{j}\in
\Sol(I_{j})$, the \emph{migrated items} $M_{t}(\vec{S})$ at time $t$ are defined
as $M_{t}(\vec{S})=\{i\in I_{t-1}\cap I_{t}\mid S_{t-1}(i)\neq S_{t}(i)\}$. 
The total \emph{migration} $\mu(\vec{S},t)$ used until time $t$ is defined as
$\mu(\vec{S},t):=\sum_{j=1}^{t}\sum_{i\in M_{j}(\vec{S})}v_{i}$, i.\,e.~as the
sum of the sizes of the migrated items. 
Inserting an item $i$ builds up a migration potential of $v_{i}$ and migrating
$i$ has cost $v_{i}$. 
More formally, for $t=1,\ldots,|I|$, let $A_{t}=\bigcup_{j\in
  \{0,1,\ldots,t-1\}}(I_{j+1}\setminus I_{j})$ be the set of items that were
inserted until time $t$ and $D_{t}=\bigcup_{j\in
  \{0,1,\ldots,t-1\}}(I_{j}\setminus I_{j+1})$ be the set of items that departed
until time $t$. 
Let $\alg{A}$ be an online algorithm that is allowed to migrate items. 
We say that $\alg{A}$ has \emph{migration factor} $\beta$, if
\begin{align*}
  \mu(\vec{S},t)\leq \beta [\Vol(A_{t})+\Vol(D_{t})].
\end{align*}
for all $t=1,\ldots,|I|$ and all $I\in \mathcal{I}^{\mathrm{on}}$. 
Here $\vec{S}$ is the sequence of solutions produced by $\alg{A}$. 
Note that our definition of migration is amortized, i.\,e.~we can build a
potential to use later on.
This will essentially allows us to repack the complete instance from time to
time.
In contrast, in the notion of worst-case migration, the total size of all
repacked items is at most $\beta\cdot v_{t}$ at each time $t$, i.\,e.~one cannot
save up migration for later use.

If $\alg{A}$ is $(\gamma+\epsilon)$-competitive for some constant $\gamma$ and
all $\epsilon > 0$, and additionally has migration factor bounded by
$f(1/\epsilon)$ for some function $f$, we say that $\alg{A}$ is \emph{robust}. 
This is due to the fact that the migration needed only depends on the desired
quality of the solution.

\subsection{The General Framework}

Before we start looking at the specific problems, we will introduce our very
general framework. The simple algorithm presented by Feldkord \emph{et
  al.}~\cite{DBLP:conf/icalp/FeldkordFGGKRW18} can
be seen as a special case of our framework. 
By using the concept of amortized migration, we can save up repacking potential,
to be used at a later time. 
Now the design of an algorithm boils down to three basic questions. 
How do we pack arriving items? 
How do we repack arriving items? 
And at what time do we repack items? 
The main idea behind our framework for packing problems is to use general
algorithms for the first two problems: 
an online algorithm to pack arriving items and an offline algorithm for the
repacking. 
The third problem is then solved by a generic combination of the two algorithms. 
In order for this approach to work, we need different criteria for both
algorithms and the packing problem that we are trying to solve. 
A simpler version of this framework was also used Feldkord \emph{et
  al.}~\cite{DBLP:conf/icalp/FeldkordFGGKRW18} for the Bin Packing problem.

\begin{definition}
Let $\Pi$ be an minimization problem with sizes $v_{i}$. We call
$\Pi$ \emph{space related}, if $\OPT(I) \ge \Vol(I)$ instances
$I\in\Pi$. Here $\Vol(I) := \sum_{i\in I}{v_i}$ is the total size of $I$. 
\end{definition}

Intuitively, this definition captures the fact that $\Pi$ is a packing problem
that needs to pack items of a certain volume in some non-overlapping way. All
problems in this paper will be space related. To make use of this relation, we
also need online algorithms such that competitiveness not only
holds for the optimum $\OPT(I)$, but also for the volume $\Vol(I)$. 
We therefore also formally introduce this necessity. 

\begin{definition}
  Let $\Pi$ be an online minimization problem with sizes $v_{i}$ and
  $\alg{A}$ be an online algorithm for $\Pi$. We say that $\alg{A}$ is
  \emph{space related with ratio $\beta$}, if
    \begin{align*}
    \Val(S_{t})\leq \beta\cdot \Vol(I_{t})+c
    \end{align*}
  for all online instances $I\in \mathcal{I}^{\text{on}}$ and all time points $t=1,2,\ldots,|I|$. Here,
  $S_{t}$ denotes the solution produced by $\alg{A}$
  at time~$t$.
\end{definition}

Trivially, a space related algorithm with ratio $\beta$ for a space related
problem implies $\beta$-competitiveness, like Next-Fit for the Bin Packing
problem. As indicated above, we will combine such an online algorithm with
another offline algorithm. To be able to efficiently combine these two
algorithms the online algorithms need to be able to build flexibly on top of the
solution of the offline algorithm.

More formally, we require that our online
algorithm takes another optional argument $S\in \Sol(I_{t})$ describing an
existing solution to the previous instance $I_{t}$ to build upon.
\begin{definition}
  Let $\Pi$ be an online minimization problem with sizes $v_{i}$ and $\alg{A}$
  be a space related online algorithm for $\Pi$ with ratio $\beta$. Furthermore
  let $I\in \mathcal{I}^{\text{on}}$ be a static online instance (i.\,e.~it
  contains no departures),  $t < t'\leq |I|$ two time points, and $S$ be
  a solution of $I_{t}$. We say that $\alg{A}$ is \emph{flexible}, if it also
  accepts $S$ as another parameter and produces upon input $I_{t'}$ and $S$ a
  solution $S'$ with $S'(i)=S(i)$ for all $i\in I_{t}\cap I_{t'}$. 
\end{definition}

Note that we define flexibility only with regard to static online instances
where no items depart. One advantage of our framework is that we only need to
design such online algorithms, but our combined algorithm will also be able to
deal with departures. Remember that our online algorithm $\alg{A}$ is given some
solution $S$ to the previous items $I_{t}$. In order to be able to ignore the departure
of items in the combined approach, we need to guarantee that $\alg{A}$ only
introduces an error of $\beta[\Vol(I_{t'})-\Vol(I_{t})]$ when it packs instance
$I_{t'}$. 

\begin{definition}
    We say that $\alg{A}$ is \emph{flexible with ratio $\beta$}, if $\alg{A}$ is
    flexible and 
    \begin{align*}
      \alg{A}(I_{t'},S)\leq \costs(S)+\beta[\Vol(I_{t'})-\Vol(I_{t})]+c
    \end{align*}
    for some constant $c$ for all instances $I\in \mathcal{I}^{\mathrm{on}}$,
    all $t\in \{1,\ldots,|I|\}$, and all $t'\in \{t+1,\ldots,|I|\}$. 
  \end{definition}

For the problems we will address later, Bin Packing and Strip Packing,
we can simply solve the instance containing the newly arriving jobs separately
and pack the new partial solution on top of the old one. Therefore this will be
a property easily fulfilled by all online algorithms that we consider later. Note that a flexible algorithm with ratio $\beta$ is also space related with
  ratio $\beta$, as we can simply choose $t=1$ and a trivial solution $S$ for
  this instance with a single item. 

\subsection*{Combining the Algorithms}
In order to obtain a robust PTAS or a robust online algorithm that is $\gamma+\epsilon$-competitive, we will need a flexible online algorithm with a
constant ratio. For the offline algorithm we will need an offline
$\gamma+\epsilon$-approximation (in the case of $\gamma=1$, this is simply a PTAS).
Basically, our final algorithm will have the same ratios as the offline
algorithm and migration factor $O(1/\epsilon)$. In order to achieve a bounded
competitiveness, we need to migrate items at some special time points. These
time points will be determined by the total volume of items that arrived or
departed since the last such time point. At these special time points, we will
use the offline algorithm to rebuild the solution completely. In between these
points, we will only apply the online algorithm for the static case. We will
therefore define phases such that during a phase we only apply the online
algorithm.

\begin{definition}
  \label{def:phases}
Let $\Pi$ be an online minimization packing problem with sizes $v_{i}$ and
$I\in \mathcal{I}^{\mathrm{on}}$ be an online instance. 

We partition $I_{1},I_{2},\ldots$ into \emph{phases} as follows: The start time  of the first phase is
$1$. If $\tau$ is the start time of the current phase, and $t\geq
\tau$ is some time point, we define the following values: (i) the complete
volume of the instance at time $\tau$ is denoted by
$V_{\tau}=\Vol(I_{\tau})$, (ii) the items inserted since $\tau$ are defined
as $\text{Ins}_{t}=\bigcup_{i=\tau}^{t-1}(I_{i+1}\setminus I_{i})$ and its
volume is $A_{t}=\Vol(\text{Ins}_{t})$, and (iii) the items departed since
$\tau$ are defined as $\text{Dep}_{t}=\bigcup_{i=\tau}^{t-1}(I_{i}\setminus
I_{i+1})$ and its total volume is $R_{t}=\Vol(\text{Dep}_{t})$. The current
phase ends at 
the earliest point of time $\tau'>\tau$ such that $A_{\tau'}+R_{\tau'}>\epsilon
V_{\tau}$. The next phase then starts at time $\tau'$. 
\end{definition}

Basically we end a phase when the total size of items that were added or removed
exceeds an $\epsilon$ factor of the total item size at the beginning of the
phase. 
So our final algorithm basically packs inserted items using the online algorithm
until a phase ends. 
Departed items stay at the same position. 
Whenever a phase ends, we compute an offline solution for the current instance
and repack our solution completely. 
After that the next phase starts, where we again only pack with the online
algorithm.
Now we will prove that this combination of algorithms will be a robust
$\gamma+\epsilon$ algorithm for $\gamma\in O(1)$ and fixed $\epsilon >0$, when
we combine two algorithms with the right properties.

\begin{lstlisting}[caption={Framework $\alg{ALG}(\alg{A}_{\mathrm{on}}, \alg{A}_{\mathrm{off}})$},label=list:8-6,captionpos=t,float,abovecaptionskip=-\medskipamount, escapeinside={(*}{*)}]
let (*$V_{\mathrm{total}} :=0$*);
let (*$V_{\mathrm{changed}} := 0$*);
let (*$S$*) be an empty solution;
for each time (*$t$*) and arriving or departing item (*$i$*)
  if item (*$i$*) arrives do
    pack (*$i$*) according to online algorithm (*$\alg{A}_{\mathrm{on}}$*) and solution (*$S$*);
  endif
  if item (*$i$*) departs do
    leave (*$i$*) in its position in the current solution (*$S$*);
  endif
  (*$V_{\mathrm{changed}} = V_{\mathrm{changed}} + v_i$*);
  if (*$V_{\mathrm{changed}} > \epsilon V_{\mathrm{total}}$*) do
    Compute offline solution (*$S$*) for (*$I_t$*) with (*$\alg{A}_{\mathrm{off}}$*);
    Continue with new solution;
    Set (*$V_{\mathrm{total}} = \Vol(I_t) $*);
    Set (*$V_{\mathrm{changed}} = 0$*);
  endif
\end{lstlisting}

\begin{theorem}
  \label{thm:main}
  Let $\Pi$ be a minimization problem fulfilling \autoref{assumption}. 
Furthermore let $\alg{A}_{\mathrm{off}}$ be a $(\gamma+\epsilon)$-approximation algorithm for the offline
version of $\Pi$ for some constant $\gamma\in O(1)$, let $1/2 \geq \epsilon >0$,
and let $\alg{A}_{\mathrm{on}}$ be a flexible algorithm with ratio $\beta$. Then the
combination of these two algorithms, denoted with
$\alg{ALG}(\alg{A}_{\mathrm{on}}, \alg{A}_{\mathrm{off}})$, is an $(\gamma
+O(1)\beta \epsilon )$-competitive robust algorithm for $\Pi$ with amortized migration
factor $O(\frac{1}{\epsilon})$.

The running time of $\alg{ALG}$ at time point $t$ is at most
$T_{\mathrm{off}}(t)+T_{\mathrm{on}}(t)$, where $T_{\mathrm{off}}(t)$
(resp.~$T_{\mathrm{on}}(t)$) is the worst-case running time of
$\alg{A}_{\mathrm{off}}$ (resp.~$\alg{A}_{\mathrm{on}}$) on an instance of $t$
items. 
\end{theorem}

\begin{proof}
Migration factor: Let $I_\tau$ be the instance at the start of a phase with
volume $V_\tau=\Vol(I_{\tau})$ and let $\tau'$ be the ending time of this phase.
As the phase ends at time $\tau'$, we have $A_{\tau'}+R_{\tau'}>\epsilon
V_{\tau}$ (where $A_{\tau'}$ and $R_{\tau'}$ are defined as in
\autoref{def:phases}). As we only ever migrate items at the end of a
phase, we only need to consider the amortized migration factor at these time
points. We assign the volume of all items inserted and departed during the phase
that ends at time $\tau'$ to this phase. Hence, a total migration potential of
$A_{\tau'}+R_{\tau'}$ was build up in this phase. The total volume of the
instance at this point is at most $V_{\tau}+A_{\tau'}$ and our amortized
migration factor is thus 
\begin{align*}
 \frac{V_{\tau}+A_{\tau'}}{A_{\tau'}+R_{\tau'}} \leq \frac{V_{\tau}}{A_{\tau'}+R_{\tau'}} +1 < \frac{V_{\tau}}{\epsilon V_{\tau}} +1 = O(1/\epsilon). 
\end{align*}
As the phases are disjoint, the total amortized migration factor is thus also
bounded by $O(1/\epsilon)$.

Competitiveness: To show the competitiveness of $\alg{ALG}$, assume that
\begin{align*}
  \label{eq:off}
  \tag{$*_{\mathrm{off}}$}
  \alg{A}_{\mathrm{off}}(I_{t})\leq (\gamma+\epsilon)\cdot \OPT(I_{t})+c_{\mathrm{off}}
\end{align*}
for some constant $c_{\mathrm{off}}$ and furthermore
\begin{align*}
  \label{eq:on}
  \tag{$*_{\mathrm{on}}$}
  \alg{A}_{\mathrm{on}}(I_{t'},S)\leq \costs(S)+\beta[\Vol(I_{t'})-\Vol(I_{t})]+c_{\mathrm{on}}
\end{align*}
for some constant $c_{\mathrm{on}}$.

Let $I\in \mathcal{I}^{\mathrm{on}}$ be some online instance and $t\in
\{1,\ldots,|I|\}$. We distinguish whether $t$ is the end  of a phase or in
the middle of a phase. If $t=\tau$ is the end of a phase, we use the offline
algorithm $\alg{A}_{\mathrm{off}}$ and thus have
\begin{align*}
  \alg{ALG}(I_{t})=\alg{A}_{\mathrm{off}}(I_{t})\leq  (\gamma+\epsilon)\cdot \OPT(I_{t})+c_{\mathrm{off}},
\end{align*}
where the inequality follows from \cref{eq:off}. 
Now consider any point of time $t$ during the phase starting at $\tau$. Like
above, let $A_t,R_t$ denote the total volumes of arrived and removed items in
this phase up time $t$. Note that this time $A_t+R_t \le \epsilon V_{\tau}$,
since otherwise we would repack.

\begin{claim}
  \label{claim:opt_tau}
  We have $\OPT(I_{\tau})\leq \OPT(I_{t})+\beta \epsilon V_{\tau}+c_{\mathrm{on}}$.
\end{claim}
\begin{proof}
  By \autoref{assumption}, the value $\OPT(I_{t})$ is minimal if only departures
  happened. We can thus assume 
  w.\,l.\,o.\,g.~that up till time $t$ some items were removed and no new items arrived. Consider an
  optimal solution $S_t$ for the instance $I_t$, where items departed and let
  $D_{t}=I_{\tau}\setminus I_{t}$ be the set of departed items. Let
  $d_{1},\ldots,d_{k}$ be some total ordering of $D_{t}$. 
  We now construct
  a new online instance $I'$ that somehow reverses the removal of $D_{t}$. The
  instance $I'$ is of length $t+k$ with $I'_{i}=I_{i}$ for
  $i\leq t$. For $i > t$, we define $I'_{i}=I_{i-1}\cup \{d_{i}\}$. We now use
  the online algorithm $\alg{A}_{\mathrm{on}}$ on this instance $I'$ with
  solution $S_{t}$ for instance $I_{t}$. At the end of instance $I'$,
  \cref{eq:on} implies that $\alg{A}_{\mathrm{on}}$ gives a feasible solution to
  $I'_{t+k}=I_{t}\cup D_{t}=I_{\tau}$ with
  \begin{align*}
    \alg{A}_{\mathrm{on}}(I'_{t+k},S_{t})\leq &\costs(S_{t})+\beta[\Vol(I_{t+k})-\Vol(I_{t})]+c_{\mathrm{on}}=\\
    &\OPT(I_{t})+\beta[\Vol(I_{\tau})-\Vol(I_{t})]+c_{\mathrm{on}}=\\
    &\OPT(I_{t})+\beta \Vol(D_{t})+c_{\mathrm{on}}\leq\\
    &\OPT(I_{t})+\beta \epsilon V_{\tau}+c_{\mathrm{on}}. 
  \end{align*}
  The last inequality follows from the fact that
  $R_{t}=\Vol(D_{t})$ by definition and $R_{t}\leq \epsilon V_{\tau}$ by
  assumption. As $\alg{A}_{\mathrm{on}}$ produces a feasible solution for
  $I'_{t+k}=I_{\tau}$, we clearly have $\OPT(I_{\tau})\leq
  \alg{A}_{\mathrm{on}}(I'_{t+k},S_{t})$ thus proving the claim.  
\end{proof}

Since the problem $\Pi$ is space related, we can conclude that 
\begin{align*}
  \label{eq:opt_vol}
  \tag{*}
\OPT(I_t) \geq \Vol(I_{t}) =   V_{\tau} +A_{t} -R_{t} \geq V_{\tau} -R_t \ge V_{\tau} - \epsilon V_{\tau} \ge \frac{V_{\tau}}{2}.
\end{align*}
Hence $V_{\tau}\leq 2\OPT(I_{t})$.

Let $S_{\tau}$ be the solution produced by $\alg{A}_{\mathrm{off}}$ at time
$\tau$ that the online algorithm $\alg{A}_{\mathrm{on}}$ is building upon.
Note that we do only remove the departed item at a time point where the offline
algorithm is used. Hence, at time $t$, the online algorithm does not produce a
solution for the instance $I_{t}$, where the departed items are already
removed. The algorithm rather works on the instance that still contains all
items that have departed since time $\tau$. This instance is denoted by
$I'_{t}$. 
We thus have
\begin{align*}
  &\alg{ALG}(I_{t}) = \alg{A}_{\mathrm{on}}(I'_{t},S_{\tau})\leq  \shortrnote{\cref{eq:on}}\\
  &\Val(S_{\tau})+\beta[\Vol(I'_{t})-\Vol(I_{\tau})]+c_{\mathrm{on}}= \shortrnote{$S_{\tau}$ produced by $\alg{A}_{\mathrm{off}}$}\\
&\alg{A}_{\mathrm{off}}(I_{\tau})+\beta[\Vol(I'_{t})-\Vol(I_{\tau})]+c_{\mathrm{on}}\leq \shortrnote{\cref{eq:off}}\\
&(\gamma+\epsilon)\cdot \OPT(I_{\tau})+c_{\mathrm{off}}+\beta[\Vol(I'_{t})-\Vol(I_{\tau})]+c_{\mathrm{on}}\leq \shortrnote{\autoref{claim:opt_tau}}\\
  &(\gamma+\epsilon)\cdot [\OPT(I_{t})+\beta \epsilon V_{\tau}+c_{\mathrm{on}}]+c_{\mathrm{off}}+\beta[\Vol(I'_{t})-\Vol(I_{\tau})]+c_{\mathrm{on}}=\\
  &(\gamma+\epsilon)\cdot [\OPT(I_{t})+\beta \epsilon V_{\tau}+c_{\mathrm{on}}]+c_{\mathrm{off}}+\beta A_{t}+c_{\mathrm{on}}\leq\shortrnote{$A_{t}+R_{t}\leq \epsilon V_{\tau}$}\\
  &(\gamma+\epsilon)\cdot [\OPT(I_{t})+\beta \epsilon V_{\tau}+c_{\mathrm{on}}]+c_{\mathrm{off}}+\beta \epsilon V_{\tau}+c_{\mathrm{on}} =\\
  &(\gamma+\epsilon) \OPT(I_{t})+(\gamma+\epsilon+1)\beta\epsilon V_{\tau}+(\gamma+\epsilon+1)c_{\mathrm{on}}+c_{\mathrm{off}}\leq\shortrnote{\cref{eq:opt_vol}}\\
  &(\gamma+\epsilon)\OPT(I_{t})+2(\gamma+\epsilon+1)\beta\epsilon \OPT(I_{t})+(\gamma+\epsilon+1)c_{\mathrm{on}}+c_{\mathrm{off}}=\\
  &(\gamma+\epsilon+2(\gamma+\epsilon+1)\beta\epsilon)\OPT(I_{t})+(\gamma+\epsilon+1)c_{\mathrm{on}}+c_{\mathrm{off}}. 
\end{align*}
As $\gamma$, $c_{\mathrm{on}}$, and $c_{\mathrm{off}}$ are constants, the last term can be written as
\begin{align*}
&(\gamma+\epsilon+2(\gamma+\epsilon+1)\beta\epsilon)\OPT(I_{t})+(\gamma+\epsilon+1)c_{\mathrm{on}}+c_{\mathrm{off}}\leq \\
&(\gamma+O(1)\cdot \beta\epsilon))\OPT(I_{t})+O(1). 
\end{align*}

The running time bound follows easily from the fact that we essentially only use
$\alg{A}_{\mathrm{on}}$ or $\alg{A}_{\mathrm{off}}$ at any given time. 
\end{proof}

\section{$2$-Dimensional Strip Packing}
In the \emph{online Strip Packing} problem, we are given a two-dimensional strip
of width $1$ and infinite height. 
At time $t$, either a rectangle $r_{t}$ with width $w(r_{t})\leq 1$ and height
$h(r_{t})\leq 1$ is inserted and needs to be packed into this strip or a
rectangle $r_{t}$ is removed from the strip. 
A packing is valid if no two rectangles intersect.
The size $v(r)$ of a rectangle $r$ is defined as $v(r)=h(r)\cdot w(r)$. 
We first focus on the case that rectangles are not allowed to be rotated. 
A simple adaption of our algorithm also handles the case, where rotations by 90
degree are allowed. 
In both cases, the goal is to minimize the height of the produced packing. 
This problem has been studied intensively in the online setting (see for example
the works cited in \cite{christensen2017survey}). 
Jansen \emph{et al.}~\cite{jansen2017strip} studied the static case in the
migration scenario, where rectangles can only arrive.

To use our framework, we need the following ingredients:
\begin{enumerate}[label=\roman*)]
\item We need to show that the Strip Packing problem is space related;
\item we need to construct a flexible online algorithm 
  with ratio $\beta$;
\item We need to construct an offline approximation algorithm. 
\end{enumerate}
Concerning the first point, it is well-known that $\OPT(I_{t})\geq \Vol(I_{t})$,
as the rectangles are not allowed to intersect and the width of the strip is
exactly $1$.

\begin{remark}
  The Strip Packing problem is space related.
\end{remark}

We will now present a flexible online algorithm with ratio $\beta=4$. 
This algorithm is a  simple adaption of the \emph{shelf algorithms} presented by
Baker and Schwarz~\cite{DBLP:journals/siamcomp/BakerS83}.
In the notion of Csirik and Woeginger~\cite{DBLP:journals/ipl/CsirikW97}, this
algorithm would be denoted as $\textsc{SHELF}(\mathsf{FirstFit}, 1/2)$. 
For the sake of completeness, we give a self-contained description and analysis. 
We first define several types of containers. 
A container $c$ of type $\gamma_{0}$ has width $1$ and height $h(c)=1$. 
For $i\in\mathbb{N}_{\geq 1}$, a container of type $\gamma_{i}$ has width $1$
and height $h(c)=2^{-i+1}$. 
For each $i\in \mathbb{Z}_{\geq 0}$, we will have at most one \emph{active}
container of type $\gamma_{i}$. 
For all other containers of this type~--~which we call \emph{closed}~--~we will
guarantee that at least $1/4$ of their volume is used by items.

We perform the following operation whenever
a new rectangle $r_{t}$ arrives:
\begin{compactitem}
\item If $w(r_{t}) \geq 1/2$, check whether a container of type $\gamma_{0}$
  exists. 
  If not, open a new active container of type $\gamma_{0}$ and place $r_{t}$
  into it. 
  If such a container $c$ already exists and $h(r_{t})+\sum_{r\in c}h(r) > 1$,
  declare $c$ as closed and open a new active container of type $\gamma_{0}$. 
  Otherwise ($h(r_{t})+\sum_{r\in c}h(r) \leq 1$), put $r_{t}$ on top of the top
  item in $c$.
\item If $w(r_{t})\leq 1/2$ and $h(r_{t})\in (2^{-i},2^{-i+1}]$, check whether a
  container of type $\gamma_{i}$ exist. If not, open a new active container of
  type $\gamma_{i}$ and place $r_{t}$ into it. If such a container $c$ already
  exists and $w(r_{t})+\sum_{r\in c}w(r) > 1$, declare $c$ as closed and open a
  new active container of type $\gamma_{i}$. Otherwise ($w(r_{t})+\sum_{r\in
    c}w(r) \leq 1$), put $r_{t}$ right to the right-most item in $c$. 
\end{compactitem}

We have at most one active container of type $\gamma_{i}$ for each $i\in
\mathbb{Z}_{\geq 0}$: we only open a new active container if we
simultaneously declare another container of the same type as closed.

We also have the following simple lemma showing that the volume of closed
containers can be bounded.
\begin{lemma}
  \label{lem:strip:volume_bound}
  Let $t$ be any time point, $i\in \mathbb{Z}_{\geq 0}$, and
  $c_{1},\ldots,c_{k}$ be the containers of type $\gamma_{i}$ at time $t$. We
  have
  \begin{align*}
    \sum_{j=1}^{k}\sum_{r\in c_{j}}v(r)\geq 1/4\bigl( \sum_{j=1}^{k}h(c_{j}) \bigr)-2^{-i+1}.
  \end{align*}
\end{lemma}
\begin{proof}
  If $k=1$, only one container of type $\gamma_{i}$ exists. As $h(c_{j})\leq
  2^{-i+1}$, the inequality holds. The left hand side is strictly larger than
  $0$ and the right hand side is strictly smaller than $0$.

  If $k > 1$, we can argue about the contained volume. If $i=0$, we have opened
  container $j>1$, as $\bigl[ \sum_{r\in c_{j-1}}h(r) \bigr]+\bigl[
  \sum_{r\in c_{j}}h(r) \bigr] > 1$. As all of these rectangles were placed in a
  container of type $\gamma_{0}$, they have width at least $1/2$. Summing up,
  as $h(c_{j})=1$, this gives
  \begin{align*}
    &\sum_{j=1}^{k}\sum_{r\in c_{j}}v(r)\geq \sum_{j=0}^{\lfloor k/2 \rfloor-1}\bigl[  \sum_{r\in c_{2j+1}}v(r)\bigr]+\bigl[  \sum_{r\in c_{2j+2}}v(r)\bigr]=\\
    &\sum_{j=0}^{\lfloor k/2 \rfloor-1}\bigl[  \sum_{r\in c_{2j+1}}h(r)\cdot \underbrace{w(r)}_{\geq 1/2}\bigr]+\bigl[  \sum_{r\in c_{2j+2}}h(r)\cdot \underbrace{w(r)}_{\geq 1/2}\bigr]\geq\\
    &(1/2)\sum_{j=0}^{\lfloor k/2 \rfloor-1}\underbrace{\bigl[  \sum_{r\in c_{2j+1}}h(r)\bigr]\bigl[  \sum_{r\in c_{2j+2}}h(r)\bigr]}_{> 1}> (1/2)[\lfloor k/2 \rfloor]\geq (1/2) (k/2 -1)=\\
    &(1/4)k - (1/2) = (1/4)\sum_{j=1}^{k}h(c_{j})-(1/2)> (1/4)\sum_{j=1}^{k}h(c_{j})-\underbrace{2}_{=2^{-0+1}}. 
  \end{align*}
If $i> 0$, we have opened container
  $c_{j}$ for $j>1$, as $\bigl[ \sum_{r\in c_{j-1}}w(r) \bigr]+\bigl[
  \sum_{r\in c_{j}}w(r) \bigr] > 1$. As all of these rectangles were placed in a
  container of type $\gamma_{i}$, they have height at least $2^{-i}$. Summing
  up, as $h(c_{j})=2^{-i+1}$, this gives
  \begin{align*}
    &\sum_{j=1}^{k}\sum_{r\in c_{j}}v(r)\geq \sum_{j=0}^{\lfloor k/2 \rfloor-1}\bigl[  \sum_{r\in c_{2j+1}}v(r)\bigr]\bigl[  \sum_{r\in c_{2j+2}}v(r)\bigr]=\\
    &\sum_{j=0}^{\lfloor k/2 \rfloor-1}\bigl[  \sum_{r\in c_{2j+1}}\underbrace{h(r)}_{> 2^{-i}}\cdot w(r)\bigr]+\bigl[  \sum_{r\in c_{2j+2}}\underbrace{h(r)}_{> 2^{-i}}\cdot w(r)\bigr]\geq\\
    &2^{-i}\sum_{j=0}^{\lfloor k/2 \rfloor-1}\underbrace{\bigl[ \sum_{r\in c_{2j+1}}w(r) \bigr]+\bigl[ \sum_{r\in c_{2j+2} }w(r) \bigr]}_{>1} > \\
    &2^{-i} \lfloor k/2 \rfloor \geq 2^{-i} (k/2 -1) = (1/4)2^{-i+1}\cdot k - 2^{-i}=\\
    &(1/4)\sum_{j=1}^{k}h(c_{j}) - 2^{-i}\geq (1/4)\sum_{j=1}^{k}h(c_{j}) - 2^{-i+1}. \qedhere
  \end{align*}
\end{proof}

 If we are given a previous packing $S$, we simply build the containers on top
 of it. Clearly, this gives a solution of height $h(S)+\sum_{j=1}^{k}h(c_{j})$,
 if $c_{1},\ldots,c_{k}$ are the containers constructed by the algorithm. We
 will now use \autoref{lem:strip:volume_bound} to show that this algorithm is
 indeed a flexible online algorithm with ratio $4$.

 \begin{theorem}
   \label{thm:strip:flexible}
   The presented algorithm $\alg{A}_{SP}$ is a flexible online algorithm for Strip
   Packing with ratio $4$.
 \end{theorem}
 \begin{proof}
   The flexibility of $\alg{A}_{SP}$ follows directly due to the fact that we only
   build upon the existing packing. We will now show that $\alg{A}_{SP}$ has ratio $4$. Let $t$ be any time point, $S$ be the previous packing
   we built upon, and $c_{1},\ldots,c_{k}$ be
   the containers constructed at time $t$. The set of rectangles contained in
   the containers is denoted as $J$.  As note above, the current packing
   has height $h(S)+\sum_{j=1}^{k}h(c_{j})$.  For $i\in \mathbb{Z}_{\geq 0}$,
   let $C_{i}\subseteq \{c_{1},\ldots,c_{k}\}$ be the set of containers of type
   $\gamma_{i}$. We then have $\sum_{j=1}^{k}h(c_{j}) = \sum_{i\geq 0}\sum_{c\in
     C_{i}}h(c)$. Using \autoref{lem:strip:volume_bound} gives
   \begin{align*}
     &\sum_{i\geq 0}\sum_{c\in C_{i}}h(c)\leq \sum_{i\geq 0}\bigl[  \sum_{c\in C_{i}}h(c)\bigr]-2^{-i+3}+2^{-i+3} =  \bigl( \sum_{i\geq 0}\bigl[  \sum_{c\in C_{i}}h(c)\bigr]-2^{-i+3} \bigr)+\sum_{i\geq 0}2^{-i+3}\leq \\
     &\bigl( \sum_{i\geq 0}\bigl[  \sum_{c\in C_{i}}h(c)\bigr]-2^{-i+3} \bigr)+16=
       4\bigl( \sum_{i\geq 0}\bigl[  \sum_{c\in C_{i}}(1/4)h(c)\bigr]-2^{-i+1} \bigr)+16\underbrace{\leq}_{\text{\autoref{lem:strip:volume_bound}}}\\
     &4\bigl( \sum_{i\geq 0} \sum_{c\in C_{i}}\sum_{r\in C}v(r) \bigr)+16=4\Vol(J)+16. \qedhere
   \end{align*}
 \end{proof}

We have now shown the first two ingredients for our framework: the problem is
space related and we gave a suitable online algorithm. The final piece~--~an
offline approximation algorithm~--~is given by the asymptotic fully polynomial
time approximation scheme (AFPTAS) of Kenyon and
R{\'{e}}mila~\cite{kenyon2000strip}, which is an
$1+\epsilon$-approximation. We can thus use \autoref{thm:main} with $\gamma=1$ and $\beta=4$
to conclude the following theorem.

\begin{theorem}
  There is a robust online algorithm for the dynamic Strip Packing problem
  that is $1+\epsilon$-competitive and has amortized
  migration factor $O(1/\epsilon)$.
\end{theorem}

\subsection*{Rotations}
If rotations by 90 degree are allowed, the resulting problem is called Strip
Packing With Rotations. For an instance $I$, we denote the height of a
corresponding optimal packing by $\OPT_{R}(I)$. As the volume of a rotated
rectangle does not change, we have $\OPT_{R}(I)\geq \Vol(I)$. Similarly, the 
volume bound of \autoref{lem:strip:volume_bound} also remains true. 
We can thus conclude the following adaption of \autoref{thm:strip:flexible}.

\begin{theorem}
  The presented algorithm $\alg{A}_{SP}$ is a flexible online algorithm for Strip
   Packing With Rotations with ratio $4$.
 \end{theorem}

 Instead of using the classical AFPTAS by Kenyon and
 R{\'{e}}mila~\cite{kenyon2000strip}, we use the AFPTAS of Jansen and van
 Stee~\cite{jansen2005strip} for the case that rotations are allowed. Using
 \autoref{thm:main} with $\gamma=1$ and $\beta=4$ gives the following theorem.

\begin{theorem}
  There is a robust online algorithm for the dynamic Strip Packing With Rotations problem
  that is $1+\epsilon$-competitive and has amortized
  migration factor $O(1/\epsilon)$.
\end{theorem}
 
The best known online algorithm with migration known is due to Jansen \emph{et
  al.}~\cite{jansen2017strip}. It also is
$1+\epsilon$-competitive, but an amortized migration factor of
$O(1/\epsilon^{9}\log^{2}(1/\epsilon))$, only works for the static case (no
rectangles are removed), and cannot handle rotations.

\section{Bin Packing}

\subsection{2-Dimensional Bin Packing}
\label{sec:2d:binpacking}
In 2-D Bin Packing, each item $i$ is given by its height $h_i\leq 1$ and its
width $w_i\leq 1$. The goal is to pack these items non-overlapping into as few
unit-sized squares (called \emph{bins}) as possible. As above, we will show the following:
\begin{enumerate}[label=\roman*)]
\item We need to show that the 2-D Bin Packing problem is space related;
\item We need to construct a flexible online algorithm 
  with ratio $\beta$;
\item We need to construct an offline approximation algorithm. 
\end{enumerate}
As the rectangles are not allowed to overlap and each bin has a total volume of $1$, 2-D Bin Packing is space related.

\begin{remark}
  The 2-D Bin Packing problem is space related.
\end{remark}

We will now present a flexible online algorithm.
This algorithm is a  simple extension of the classical algorithm presented by
Coppersmith and Raghavan~\cite{coppersmith1989multidimensional}.
For the sake of completeness, we give a self-contained description and analysis. 
We categorize items as follows: We call an item
\emph{vertical} if $w_{i}\leq h_{i}$ and \emph{horizontal} if $w_i>h_i$. Note
that squares with $h_i=w_i$ will be considered as vertical. Without loss of generality
we will explain in the following how to pack vertical items. Horizontal items can and 
will be placed into separate bins with the same strategy, just altered for horizontal items. 
Note that we later also need to account for these horizontal bins.

We further assign each item $i$ a \emph{size class}. Item $i$ is in size class 
$j\in \mathbb{N}_{\ge 1}$ if $\nicefrac{1}{2^{j-1}} \ge h_i >\nicefrac{1}{2^{j}}$.
 Further we say an item $i$ is \emph{square-like}, if $i$
is in size class $j$ and furthermore $ w_i > \nicefrac{1}{2^{j}}$. The general
idea is that for every arriving item in size class $j$, we will assign a square
slot of size $\nicefrac{1}{2^{j-1}}$ in some bin. A square slot of this size is
called a \emph{slot of class $j$}. 
Note that an item of size class $j$ always fits into a slot of class $j$, as we only handle vertical items with
$w_{i}\leq h_{i}$ here. Our goal is to fill all opened slots with items
until $\nicefrac{1}{4}$ of the total area of the slot is
covered. Square-like items will immediately fill such a slot to this extent. For
the other items, we will \emph{reserve} such slots for a size class $j$ and stack items
from left to right until the total width of all items in that slot exceeds
$\nicefrac{1}{2^{j}}$. Since the height of items assigned to this slot exceed
$\nicefrac{1}{2^{j}}$ as well, $\nicefrac{1}{4}$ of the total slot will be
covered at that point. A slot can have three states: it is either (i)
\emph{empty} and thus contains no item, (ii) \emph{reserved} for class $j$ and
thus only contains items of class $j$ or (iii) \emph{closed} if at least
$\nicefrac{1}{4}$ of its total volume is filled with items.

In order to assign items to these slots, we will keep up to two open bins. 
The first open bin will hold items of size class $1$ and we use the complete bin
as a single reserved slot. 
The second open bin will receive items of size class $\ge 2$. 
Initially, this bin is split into four empty slots of class $2$.
We now define our online algorithm $\alg{A}_{2\text{-}D}$:
\begin{enumerate}
\item If a square-like item of size class $1$ arrives, we open a new bin, place
  the item in there and close it immediately. 
  By definition of size class $1$, at least $\nicefrac{1}{4}$ of the volume of
  this bin will be filled.
\item If a non-square-like item $i$ of size class $1$ arrives, we place it to
  the right of the non-square-like items in the first open bin or on the left
  end, if no such items exist. 
  As $w_{i}\leq \nicefrac{1}{2}$, we can always place this item in the bin. 
  If the sum of the widths of the items in the first bin now is at least
  $\nicefrac{1}{2}$, close this bin and open a new first bin. 
  As every item placed in the bin has height $h_{i}\geq \nicefrac{1}{2}$, at
  least $\nicefrac{1}{4}$ of the volume of this bin will be filled. 
\item If an item $i$ of class $j\geq 2$ arrives, find a non-closed slot of class
  $j'\leq j$ with maximal $j'$ in the second open bin.
  \begin{enumerate}
  \item If $j'=j$, we choose a slot reserved for $j'$ or~--~if no such slot
    exists~--~the top-most left-most of the slots of class $j'$. 
    We declare this slot as reserved for $j'$ and place the item right of the
    items placed in this slot or on the left end, if no such items exist. 
    As $w_{i}\leq h_{i}\leq \nicefrac{1}{2^{j-1}}$, we can always place $i$ in
    this slot. 
    If the sum of the widths of the items placed in this slot is at least
    $\nicefrac{1}{2^{j}}$, we close this slot.
  \item If $j' < j$, note that we can split such a slot into four equal-sized
    slots of class $j'+1$. 
    Splitting one slot will leave us with four slots: 
    three of the slots will remain empty and the last slot will either be split
    even more (if its class $j'+1$ is smaller than $j$) or we reserve this slot
    for class $j$ and put the new item into it. 
    We repeat this splitting until we have created a slot of class $j$ and put
    $i$ on the left end. 
    If the sum of the widths of the items placed in this slot is at least
    $\nicefrac{1}{2^{j}}$ (i.\,e.~if $i$ is square-like), we close this slot.
  \item If no slot of size $j' \leq j$ exists, we close the second bin and open
    a new second bin containing four empty slots of class $2$. 
    We place $i$ into this bin as above. 
  \end{enumerate}
\end{enumerate}

We will now show that every bin created by the algorithm contains a most three
empty slots of a certain class. 
\begin{lemma}
  \label{lemma:at-most-three}
  For each $j\geq 1$, every bin created by $\alg{A}_{2\text{-}D}$ contains at most three empty slots of class $j$. 
\end{lemma}
\begin{proof}
  For $j=1$, every bin contains at most a single slot of class $1$.
  
  For $j\geq 2$, assume that this is not the case and our algorithm may generate
  a slot assignment with four or more slots of the same class $j^{*}$. 
  Consider some arriving item $i$ of class $j$ such that before the arrival of
  $i$ there are at most three empty slots of class $j^{*}$ and afterwards, there
  are four or more of these slots. 
  Since the number of empty slots of class $j^{*}$ increased, our algorithm must
  have split a slot of class $j' < j^{*}$. 
  Note that this will only happen if $j^{*}\leq j$. 
  Splitting this slot of class $j'$ over and over will create three slots of
  class $j'+1,j'+2,\ldots,j^{*}-1$. 
  Finally, four slots of class $j^{*}$ are created. 
  One of this slots will either be directly used for $i$ (if $j^{*}=j$) or will
  be split into smaller slots. 
  Hence, only three new slots of class $j^{*}$ are created. 
  By assumption, the insertion of $i$ leads to at least four slots of class
  $j^{*}$. 
  Hence, there must have been a slot of class $j^{*}$ before. 
  But, by definition of the algorithm, we would rather have chosen this slot
  instead of the slot of class $j'$, as $j' < j^{*}$. 
  Hence, this is a contradiction. 
\end{proof}

To make the algorithm $\alg{A}_{2\text{-}D}$ flexible, we simply ignore all
previous used bins and only work on the newly generated bins. 

\begin{theorem}
\label{thm:2d:flexible}
  The proposed algorithm $\alg{A}_{2\text{-}D}$ for 2-D Bin Packing is a
  flexible online algorithm with ratio $\frac{48}{5}$.
\end{theorem}

\begin{proof}
  The flexibility of $\alg{A}_{2\text{-}D}$ follows directly from the fact that
  we only work on newly created bins.

  To analyze the ratio of the algorithm we look at how much total area is
  covered when a complete bin has been closed.
  First we can observe that every closed slot inside a closed bin has at least
  an overall area of $\frac{1}{4}$ covered: 
  we only close a slot if half of its width is covered and we only put items
  into a slot that cover at least half of the height.
  
  Another easy observation is the fact that the algorithm may have at most one
  reserved slot for each class, since the algorithm only opens a new reserved
  slot for a class whenever the current reserved one is closed. 
  We may open a second slot for a square-like item, but will also close it
  immediately. 
  By neglecting the items in reserved slots, the unused space due to reserved
  slots in every bin therefore can be bounded by
  \begin{align*}
    \tag{*}
    \label{eq:res-slots}
  \sum_{j=1}^{\infty}\nicefrac{1}{4^j} = \frac{1}{1-\nicefrac{1}{4}}-1 =
  \frac{1}{3}.
    \end{align*}
  By \autoref{lemma:at-most-three}, we know that every bin contains at most
  three empty slots of each size class. 
  Note that a closed bin cannot contain an empty slot of class $2$, as such a
  bin will only receive items of this size class or higher. 
  With this we can now conclude that the unused space due to empty slots in a
  closed bin is at most
  \begin{align*}
    \tag{**}
    \label{eq:empty-slots}
    3 \sum_{i=2}^{\infty}\nicefrac{1}{4^i} =3
  [\frac{1}{1-\nicefrac{1}{4}}-1-\nicefrac{1}{4}]= 3\cdot  \nicefrac{1}{12} = \nicefrac{1}{4}.
  \end{align*}
  
  By combining \cref{eq:res-slots} and \cref{eq:empty-slots}, at least a total
  area of $1 - \nicefrac{1}{4} - \nicefrac{1}{3} = \nicefrac{5}{12}$ of every
  closed bin is occupied by closed slots. 
  As discussed earlier, every such closed slot has at least $\nicefrac{1}{4}$ of
  its total area covered so we finally can conclude that at least a total area
  of $\nicefrac{5}{48}$ of every closed bin must be covered.

  The algorithm keeps at most four open bins, one for items of size class $1$,
  one for items of other size classes, and we need two respective bins for
  horizontal items containing the same classes as well. 
  Finally we get that our algorithm for an instance $I$ with total volume
  $\Vol(I)$ will use at most $\nicefrac{48}{5}\Vol(I) +4$ bins in total.
\end{proof}

Finally, we can use the approximation algorithm of Bansal and
Khan~\cite{bansal2014improved} that is an $1{.}405$-approximation
for 2-D Bin Packing. We can thus use \autoref{thm:main} with $\gamma=1{.}405$ and $\beta=\nicefrac{48}{5}$
to conclude the following theorem.

\begin{theorem}
  There is a robust online algorithm for the dynamic 2-D Bin Packing
  problem that is $1{.}405+\epsilon$-competitive and has
  amortized migration factor $O(1/\epsilon)$.
\end{theorem}

To the best of our knowledge, this is the first robust online algorithm for
dynamic 2-D Bin Packing. 
Note that the best known lower bound for the competitiveness of any online
algorithm for online 2-D Bin Packing without migration is $1{.}856$ due to Van
Vliet~\cite{van1995lower}. 

\subsection*{Rotations}
As for Strip Packing, allowing rotations of the rectangles by 90 degrees gives
rise to a problem called 2-D Bin Packing With Rotations. The corresponding
optimal number of bins needed to pack instance $I$ is denoted by $\OPT_{R}(I)$.
Rotations are invariant with regard to the volume of a rectangle and thus
$\OPT_{R}(I)\geq \Vol(I)$. We can thus again use our online algorithm
$\alg{A}_{2\text{-}D}$ to obtain the following adaption of
\autoref{thm:2d:flexible}. 

\begin{theorem}
  The proposed algorithm $\alg{A}_{2\text{-}D}$ for 2-D Bin Packing With
  Rotations is a flexible online algorithm with ratio $\frac{48}{5}$.  
\end{theorem}

The approximation algorithm of Bansal and Khan~\cite{bansal2014improved} used
above can also handle the case of rotation and thus is an
$1{.}405$-approximation for 2-D Bin Packing With Rotations. 
We can thus use \autoref{thm:main} with $\gamma=1{.}405$ and
$\beta=\nicefrac{48}{5}$ to conclude the following theorem.

\begin{theorem}
  There is a robust online algorithm for the dynamic 2-D Bin Packing With
  Rotations problem that is $1{.}405+\epsilon$-competitive and has amortized
  migration factor $O(1/\epsilon)$.
\end{theorem}

To the best of our knowledge, this is the first robust online algorithm for
dynamic 2-D Bin Packing With Rotations. 
Note that the best known lower bound for the competitiveness of any online
algorithm for 2-D Bin Packing With Rotations without migration is $1{.}6707$ due
to Heydrich and van Stee~\cite{heydrich2016hypercube}.

\subsection{$d$-Dimensional Bin Packing}
\label{sec:d:binpacking}

We will now look at the problem of packing $d$-dimensional hyperrectangles into
as few unit-sized hypercubes as possible for higher dimensions $d > 2$. 
This problem is called \emph{$d$-Dimensional Hyperrectangle Packing}. 

To obtain an suitable flexible online algorithm for this problem, we will
generalize the $2$-dimensional Bin Packing algorithm $\alg{A}_{2\text{-}D}$ from
Section~\ref{sec:2d:binpacking} The side length of dimension $i\in
\{1,\ldots,d\}$ of a $d$-dimensional hyperrectangle $r$ is denoted as $r_{i}$
and its volume is $v(r)=\prod_{i=1}^{d}r_{i}$. 
Items will be classified almost the same as before in a straight-forward
generalization. 
For a permutation $\pi\colon \{1,\ldots,d\}\to \{1,\ldots,d\}$, we associate a
set of rectangles $r$ with $r_{\pi(1)}\leq r_{\pi(2)}\leq \ldots r_{\pi(d)}$
with it.
We treat each of these permutations separately and pack the corresponding
rectangles in separate bins. 
This will only give an additive error of $d!$ in our approximation guarantee. 
In the following, we thus fix a permutation $\pi$. 
A rectangle associated with $\pi$ is in size class $j\in \mathbb{N}_{\geq 1}$,
if $\nicefrac{1}{2^{j-1}} \ge r_{\pi(d)} >\nicefrac{1}{2^{j}}$, hence the class
of a rectangle depends on its largest side length.
Rectangles of class $j$ are packed into slots that are hypercubes with side
length $\nicefrac{1}{2^{j-1}}$. 
Such a slot is called a \emph{slot of class $j$}. 
Note that an item of size class $j$ always fits into a slot of class $j$, as we
fixed the permutation $\pi$. 
We now proceed as in Section~\ref{sec:2d:binpacking}. 
If an item of class $j$ arrives, we put it into an open slot of this class. 
If no such slot exists, we split a slot of class $j'\leq j$ into $2^{d}$ slots
of class $j'+1$ recursively until an empty slot of class $j$ is created. 
We denote this straight-forward generalization of $\alg{A}_{2\text{-}D}$ as
$\alg{A}_{d\text{-}D}$. 
It is easy to see that the $d$-dimensional analogue of
Lemma~\ref{lemma:at-most-three} also holds. 

\begin{lemma}
  \label{lemma:d:at-most-three}
  For each $\pi$ and each $j\geq 1$, every bin created by $\alg{A}_{d\text{-}D}$
  contains at most $2^{d}-1$ empty slots of class $j$ associated with $\pi$. 
\end{lemma}

Every closed slot inside a closed bin has at least   an overall area of $2^{-d}$
covered, as in the proof of Theorem~\ref{thm:2d:flexible}. We thus obtain  the
following generalization of Theorem~\ref{thm:2d:flexible} by not using any bins
created by the previous solution $S$. 
\begin{theorem}
\label{thm:d:flexible}
  The proposed algorithm $\alg{A}_{d\text{-}D}$ for $d$-Dimensional Hyperrectangle Packing is a
  flexible online algorithm with ratio $\frac{2^{2d}-3\cdot
  2^{d}+1}{2^{2d}(2^{d}-1)}$.
\end{theorem}
\begin{proof}[Proof Sketch]
  By neglecting the items in reserved slots, the unused space due to reserved
  slots in every bin can be bounded by
  \begin{align*}
    \tag{*}
    \label{eq:d:res-slots}
  \sum_{j=1}^{\infty}\nicefrac{1}{(2^{d})^j} = \frac{1}{1-\nicefrac{1}{2^{d}}}-1 =
  \frac{1}{2^{d}-1}.
    \end{align*}
    By \autoref{lemma:d:at-most-three}, we know that every bin contains at most
    $2^{d}-1 $ empty slots of each size class for each $\pi$.
    With this we can now conclude that the unused space due to empty slots in a
    closed bin is at most
  \begin{align*}
    \tag{**}
    \label{eq:d:empty-slots}
    (2^{d}-1) \sum_{i=2}^{\infty}\nicefrac{1}{(2^{d})^i} =(2^{d}-1)
  [\frac{1}{1-\nicefrac{1}{2^{d}}}-1-\nicefrac{1}{2^{d}}]= (2^{d}-1)\cdot  \nicefrac{1}{(2^{d}-1)\cdot 2^{d}} = \nicefrac{1}{2^{d}}.
  \end{align*}
  By combining \cref{eq:d:res-slots} and \cref{eq:d:empty-slots}, at least a total
  area of $1 - \nicefrac{1}{2^{d}} - \nicefrac{1}{2^{d}-1} = \frac{2^{2d}-3\cdot
  2^{d}+1}{2^{d}(2^{d}-1)}$ of every
  closed bin is occupied by closed slots. 
  As discussed earlier, every such closed slot has at least $\nicefrac{1}{2^{d}}$ of
  its total area covered so we finally can conclude that at least a total area
  of $\frac{2^{2d}-3\cdot
  2^{d}+1}{2^{2d}(2^{d}-1)}$ of every closed bin must be covered.
\end{proof}

The best known offline algorithm for the $d$-dimensional Hyperrectangle Packing
problem has ratio $1{.}69103^{d-1}$ and is due to
Caprara~\cite{caprara2002harmonic} (see also~\cite{christensen2017survey}).
We can thus use \autoref{thm:main} with $\gamma=1{.}69103^{d-1}$ and
$\beta=\frac{2^{2d}-3\cdot
  2^{d}+1}{2^{2d}(2^{d}-1)}$ to conclude the following theorem.

\begin{theorem}
  There is a robust online algorithm for the dynamic $d$-dimensional
  Hyperrectangle Packing problem that is $1{.}69103^{d-1}+\epsilon$-competitive
  and has amortized migration factor $O(1/\epsilon)$.
\end{theorem}

\subsubsection{Hypercube Packing}
If every hyperrectangle is in fact a hypercube, we obtain a substantially easier
problem, called \emph{$d$-Dimensional Hypercube Packing}. 
The resulting algorithm is also much easier to describe. 
Here, a slot of class $j$ is a hypercube of side length $\nicefrac{1}{2^{j-1}}$.
Putting an item of class $j$ into a slot of class $j$ covers at least a fraction
of $\nicefrac{1}{2^{d}}$ of the total volume of the slot. 
Hence, slots are either empty or closed and do not need to be reserved. 
Every $d$-dimensional hypercube with side length $s$ can be divided into $2^d$
equal sized hypercubes with side length $\nicefrac{s}{2}$. 
The slot assignment happens like above: 
We start with single empty active bin containing $2^{d}$ empty slots of class
$2$.
An item of class $1$ is simply put into its own bin. 
Whenever an item $i$ of class $j\geq 2$ arrives, we first try to find an empty
slot of class $j$. 
If any such slot exists, put $i$ into it and close the slot. 
If no slot of class $j' \leq j$ exists, we close the bin and open a new bin to
insert $i$ into. 
If a slot of class $j' < j$ exists, we split the smallest such slot that is
still larger than the required slot size until an empty slot of size $j$ is
created.

We call this algorithm $\alg{A}_{d\text{-hyper}}$ and again make it
flexible by simply opening new bins. 
Similar to the two-dimensional case, we can show that there can't be too many
empty slots of the same class. The proof is just an adaption of the proof of
\autoref{lemma:at-most-three}. 

\begin{lemma}
  \label{lemma:at-most-2d-minus}
  For each $j\geq 1$, every bin created by $\alg{A}_{d\text{-hyper}}$ contains at most $2^{d}-1$ empty slots of class $j$. 
\end{lemma}
\begin{proof}
  For $j=1$, every bin contains at most a single slot of class $1$.
  
  For $j\geq 2$, assume that this is not the case and our algorithm may generate
  a slot assignment with $2^{d}$ or more slots of the same class $j^{*}$. 
  Consider some arriving item $i$ of class $j$ such that before the arrival of
  $i$ there are at most $2^{d}-1$ empty slots of class $j^{*}$ and afterwards,
  there are $2^{d}$ or more of these slots. 
  Since the number of empty slots of class $j^{*}$ increased, our algorithm must
  have split a slot of class $j' < j^{*}$. 
  Note that this will only happen if $j^{*}\leq j$. 
  Splitting this slot of class $j'$ over and over will create $2^{d}-1$ slots of
  class $j'+1,j'+2,\ldots,j^{*}-1$. 
  Finally, $2^{d}$ slots of class $j^{*}$ are created. 
  One of this slots will either be directly used for $i$ (if $j^{*}=j$) or will
  be split into smaller slots. 
  Hence, only $2^{d}-1$ new slots of class $j^{*}$ are created. 
  By assumption, the insertion of $i$ leads to at least $2^{d}$ slots of class
  $j^{*}$. 
  Hence, there must have been a slot of class $j^{*}$ before. 
  But, by definition of the algorithm, we would rather have chosen this slot
  instead of the slot of class $j'$, as $j' < j^{*}$. 
  Hence, this is a contradiction. 
\end{proof}

\begin{theorem}
  The proposed algorithm $\alg{A}_{d\text{-hyper}}$ for $d$-Dimensional
  Hypercube Packing is a flexible online algorithm with ratio
  $\frac{2^{2d}}{2^d-1}$.
\end{theorem}

\begin{proof}
  The flexibility follows from the fact that we do not touch the already used
  bins.
  
  Just like in the two dimensional case we analyze the free space in a closed
  bin. As observed above, every closed slot has at least an overall area of
  $\nicefrac{1}{2^{d}}$ of its size covered. We now take a look at the empty
  slots in a closed bin. 
  
  In bins that were closed due to the insertion of an item of size class $1$, at
  least a fraction of $\nicefrac{1}{2^{d}}$ of the bins area is covered. 
  Hence, consider a bin that received items with side length less or equal than
  $\nicefrac{1}{2}$. 
  As observed above, every assigned slot in this bin is either closed or empty.
  \autoref{lemma:at-most-2d-minus} guarantees that we have at most $2^{d}-1$
  empty slots of any given class. 
  Since there can not be any reserved slots, we can already calculate how much
  space we lose in total, by looking at empty slots.
  Note again that all slots of class $2$  will be used by some items
  before the overall bin gets closed. 
  Therefore the total volume assigned to empty slots can be at most
  \begin{align*}
  &(2^d-1) \sum_{i=2}^{\infty}\frac{1}{\bigl( 2^{d} \bigr)^{i}} =
  (2^d-1)[\frac{1}{1-2^{-d}}-1-\frac{1}{2^{d}}] =
    (2^{d}-1)[\frac{2^{d}}{2^{d}-1}-1-\frac{1}{2^{d}}]=\\
    &2^{d}-2^{d}+1-1+\frac{1}{2^{d}}=\frac{1}{2^{d}}.
  \end{align*}
  
  We get that at least $1- \nicefrac{1}{2^d} = \nicefrac{2^d-1}{2^d}$ of the
  total area of a closed bins is used by closed slots. As each of these slots
  covers at least an area of $\nicefrac{1}{2^{d}}$ of its size, at least
  a total volume of $\nicefrac{2^d-1}{2^d}\cdot  \nicefrac{1}{2^d}
  =\nicefrac{2^d-1}{2^{2d}} $ of a closed bin is covered. We also have to account that our algorithm may keep one open bin, to place items. Together we can conclude that our algorithm uses at most $\nicefrac{2^{2d}}{2^d-1}\Vol(I) +1$ bins in total.
\end{proof}

Finally, we can use the offline APTAS from Bansal \emph{et
  al.}~\cite{bansal2006binpacking}. Using \autoref{thm:main} with $\gamma=1$ and
$\beta=\nicefrac{2^{2d}}{2^d-1}$ allows the conclusion of the following theorem.

\begin{theorem}
  For every $d\geq 2$, there is a robust online algorithm for the  dynamic
  $d$-Dimensional Hypercube Packing problem that is
  $1+\epsilon$-competitive and has amortized migration factor $O(1/\epsilon)$. 
\end{theorem}

In contrast, the best known online algorithm with migration for the
$d$-Dimensional Hypercube Packing is due to Epstein and
Levin~\cite{epstein2013cubepacking}. 
It is also $1+\epsilon$-competitive, but can only handle the static case and has
migration factor $(1/\epsilon)^{\Omega(d)}$. 
Note however that they use worst-case migration, i.\,e.~they are not allowed to
repack the complete instance every once in a while but need to make slight
adaptions carefully throughout the run of the algorithm.

\section{$d$-Dimensional Strip Packing}
In the $d$-dimensional version of online strip packing, called the
\emph{$d$-Dimensional Strip Packing} problem, we are given a $d$-dimensional
cuboid that has size $1$ in $d-1$ dimensions and infinite size in the last
dimension. 
Even in the general case with more dimensions we will consider this last
dimension as \emph{height}. 
At each point of time a $d$-dimensional hyperrectangle $r$ arrives. 
The size of its $i$th dimension is denoted by $r_{i}\in (0,1)$ and its volume is
$v(r)=\prod_{i=1}^{d}r_{i}$. 
The task is again to pack these cuboid with no intersection such that the height
is minimized. 
Like above we denote with $I_{t}$ the set of rectangles present at time $t$,with
$\Vol(I_{t})$ the total volume and with $\OPT(I_{t})$ the optimal height.
This version of Strip Packing is also space related, since the base of the
packing space has side lengths of $1$.

\begin{remark}
  The $d$-dimensional Strip Packing problem is space related.
\end{remark}
It leaves to show that there are respective online and offline algorithms for
our framework. 
For a flexible online algorithm we will generalize our above approach for the
$d$-dimensional case in a straight-forward way. 
A \emph{container} $c$ of type $\gamma_{i}$ is a $d$-dimensional hyperrectangle
with $c_{1}=c_{2}=\ldots=c_{d-1}=1$ and $c_{d}=2^{-i}$. 
For each $i\in \mathbb{Z}_{\geq 0}$, we will have at most one \emph{active}
container of type $\gamma_{i}$. 
For all other containers of this type~--~which we call \emph{closed}~--~we will
guarantee that at least a constant fraction of their volume is used by items. 
We assign a hyperrectangle $r$ of height $r_{d}\in (2^{-i-1},2^{-i}]$ to a
container of type $\gamma_{i}$. 
We then treat these hyperrectangles as $d-1$-dimensional hyperrectangles by
projecting to its first $d-1$ coordinates and also treat the container as a
$d-1$-dimensional hypercube. 
We then pack the projected hyperrectangles into the projected hypercube with the
algorithm $\alg{A}_{(d-1)\text{-}D}$ described in
Section~\ref{sec:d:binpacking}. 
Theorem~\ref{thm:d:flexible} then guarantees that a fraction of
$\frac{2^{2(d-1)}(2^{d-1}-1)}{2^{2(d-1)}-3\cdot 2^{d-1}+1}$ of each projected
container is filled. 
As the original, non-projected containers have height $2^{-i}$ and every
non-projected hyperrectangle has height at least $2^{-i-1}$, we lose a factor of
$2$ here. 
The resulting algorithm, called $\alg{A}_{d-SP}$, is thus a flexible online
algorithm and we obtain the following generalization of
Theorem~\ref{thm:strip:flexible}.

\begin{theorem}
   \label{thm:strip:d:flexible}
   The presented algorithm $\alg{A}_{d-SP}$ is a flexible online algorithm for
   $d$-dimensional Strip
   Packing with ratio $O(\frac{2^{2d}-3\cdot
  2^{d}+1}{2^{2d}(2^{d}-1)})$. 
 \end{theorem}

 To the best of our knowledge, there is no work that explicitly deals with the
 construction of approximation algorithms for the $d$-dimensional Strip Packing
 problem. As shown above, any such result can be used in the context of our
 framework to obtain a robust online algorithm with corresponding competitive
 ratio. 

\subsection{Hypercube Strip Packing}
In the following we will restrict this problem to the case where each item is a
hypercube, so we have that each entry $r_{i}$ is the same. 
In this version, that is also known as \emph{online Hypercube Strip Packing}, we
will denote this value with $s(r)$ and call it the \emph{side length} of the
hypercube $r$. 
The size of this hypercube is given by its volume, which is $v(r) =s(r)^d$. 
Note that in this variant allowing rotations by 90 degree makes no difference.

In fact the approach becomes a little easier due to our restriction to
hypercubes. 
Again we define containers: 
For $i\in\mathbb{N}_{\geq 1}$, a container $c$ of type $\gamma_{i}$ has size $1$
in every dimension except the height and height of $h(c)=2^{-i+1}$. 
The idea is the same as before: 
Keep at most one \emph{open} container of each size, while we make sure at least
$2^{-d}$ of the total area of each \emph{closed} bin is covered. 

The assignment of items to containers works the same as well: 
Whenever a hypercube $c_t$ with side length $s(c_t) \in (2^{-i},2^{-i+1}]$
arrives , we will try to pack it into an open container of type $\gamma_i$.
If there is no open container of this type or it does not fit, we open a new one
and close the old one.
In order to place these hypercubes inside of the container we subdivide every
container into equal sized cubic slots.
Given a container of type $\gamma_i$, we further divide it into hypercubes with
side length $2^{-i+1}$, which we call \emph{slots}. 
Now when a hypercube $c_t$ with $s(c_t)> 2^{-i}$ gets assigned to a container,
we pack it into one of the free slots, such that it does not intersect any other
slots. 

Overall this yields a space related online algorithm. The flexibility follows like in the 2-D case from the fact, that we can simply extend any existing packing.

\begin{theorem}
   \label{thm:hypercubestrip:flexible}
   The presented algorithm $\alg{A}_{SP\text{-hyper}}$ is a flexible online algorithm for $d$-dimensional Hypercube Strip
   Packing restricted to hypercubes with ratio $2^{d}$.
 \end{theorem}

\begin{proof}
The flexibility follows again from the properties of the Strip Packing problem, given a packing of some items $S$, we can simply ignore that packing and put newly arriving items on top of the existing packing.

First off we show that each closed container has at least $2^{-d}$ of its total volume covered. Let $c$ be a closed container of type $\gamma_i$ with volume  $v(c) = h(c)$, since the size in every dimension except height is exactly $1$. Note that we subdivided $c$ into multiple cubic slots with total volume $h(c)^{d} = 2^{d(-i+1)}$. Since $c$ is closed each slot got assigned an item $r$ with volume 
$v(r) = s(r)^d > 2^{-i \cdot d} = 2^{-d} \cdot  2^{-d \cdot i+d} = 2^{-d} \cdot h(c)^{d}$. Since the slots divide the container completely and every slot has $2^{-d}$ of its total volume covered, 
we can conclude that $\sum_{r\in c}v(r)\geq (2^{-d}) h(c)$.

Now we will look at the height of a flexibly created solution. Let $t$ be a point of time, $S$ be the previous packing we extended and let $c_1,..,c_k$ be the containers created up to time $t$. Let $O$ be the set of open containers and $L$ the set of closed containers. Let $S_t$ be the result of our algorithm at time $t$. Let $l$ be the type of the smallest open container in our current solution. We then can observe for the height of open containers that:
$\sum_{c\in L}{h(c)} \le \sum_{i=1}^{l}{2^{-i+1}} = \sum_{i=0}^{l}{2^{-i}} \le 2 $.

Finally we have for the total height of our solution that 
\begin{align*}
&h(S_t) = h(S) + \sum_{c\in L}{h(c)}+ \sum_{c\in O}{h(c)} \leq h(S) + \sum_{c\in L}{ 2^{d} \sum_{r\in c}{v(r)}}+ \sum_{c\in O}{h(c)} \le \\
& h(S) + 2^{d} \sum_{c\in L}{  \sum_{r\in c}{v(r)}}+ 2 \le  h(S) + 2^{d} \Vol(I_{t})+ 2\qedhere
\end{align*}
\end{proof}

With this we know our problem is space-related and has an appropriate space related online algorithm. As for the offline algorithm, we will use a result from Harren, who gave an APTAS for the Hypercube Strip Packing problem ~\cite{harren2009hypercube}. By using
 \autoref{thm:main} with $\gamma=1$ and $\beta=2^d$ we get the following theorem.

\begin{theorem}
  There is a robust online algorithm for the dynamic $d$-dimensional Hypercube Strip Packing problem
  that is $1+\epsilon$-competitive and has amortized
  migration factor $O(1/\epsilon)$.
\end{theorem}

\section{Vector Packing}
In the \emph{online $d$-dimensional Vector Packing} problem, at time $t$ either
a vector $w_{t}\in (\mathbb{Q}\cap [0,1])^{d}$ is inserted and needs to be
packed or is removed. 
The size $v(w_{t})$ of such a vector $w_{t}=(w[1],\ldots,w[d])$ is defined as
the average sum of its components, i.\,e.~$v(w_{t})=\sum_{j=1}^{d}w[j]/d$.
The goal is to pack these vectors into as few as possible bins as possible.
Here, a \emph{bin} $B$ is a subset of vector such that $\sum_{w\in B}w[j]\leq 1$
for $j=1,\ldots,d$.

To use our framework, we need the following ingredients:
\begin{enumerate}[label=\roman*)]
\item We need to show that the  problem is space related;
\item we need to construct a flexible online algorithm 
  with ratio $\beta$;
\item We need to construct an offline approximation algorithm. 
\end{enumerate}

As each bin can contain items of volume at most $1$, it is easy to see that
the $d$-dimensional Vector Packing problem is space related.
\begin{remark}
  The $d$-dimensional Vector Packing problem is space related.
\end{remark}

We will now present a flexible online algorithm with ratio $\beta=2d$ that is a
simple adaption of the well-known next fit online algorithm for bin packing.
Every bin will have an index to guarantee a linear ordering. 
Whenever a vector $w$ arrives, we first check whether $w$ can be packed into an
existing bin. 
If this is possible, we add $w$ to such a bin with minimal index.
If no such bin exists, we open a new bin containing $w$.
If we are given a previous packing $S$, we simply ignore the previous bins and
do not put any vector in them. 

\begin{theorem}
  \label{thm:vector:flexible}
  For every $d\geq 1$, the presented algorithm $\alg{A}_{VP}$ is a flexible
  online algorithm for $d$-dimensional Vector Packing with ratio $2d$. 
\end{theorem}
\begin{proof}
  The flexibility of $\alg{A}_{VP}$ follows directly due to the fact that we
  only build upon the existing packing. 
  We will now show that $\alg{A}_{VP}$ has ratio $2d$. 
  Let $t$ be any time point, $S$ be the previous packing we built upon, and
  $B_{1},\ldots,B_{k}$ be the bins opened by $\alg{A}_{VP}$. 
  Let $J$ be the set of vectors packed into $B_{1},\ldots,B_{k}$. 
  The current packing then uses $|S|+k$ bins, where $|S|$ denotes the number of
  bins used by the previous packing $S$.
  Consider two adjacent bins $B_{i}$ and $B_{i+1}$ for $i\in \{1,\ldots,k-1\}$.
  We have opened bin $B_{i+1}$, because there is some vector $w'\in B_{i+1}$ and
  some index $j\in \{1,\ldots,d\}$ such that $\bigl[ \sum_{w\in B_{i}} w[j]
  \bigr]+w'[j] > 1$.
  Denoting this index with $j=j(i)$, we have $\sum_{w\in B_{i}\cup
    B_{i+1}}w[j(i)]>1$ and thus $d\cdot \Vol(B_{i}\cup B_{i+1}) > 1$. 
  Hence, we can conclude
  \begin{align*}
    &k  < 1+\sum_{i=1}^{k-1}d\cdot \Vol(B_{i}\cup B_{i+1})  = 1+d\cdot \sum_{i=1}^{k-1}(\Vol(B_{i})+\Vol(B_{i+1}))=\\
    &1+d\cdot \bigl[  (\sum_{i=1}^{k-1}\Vol(B_{i}))+(\sum_{i=2}^{k}\Vol(B_{i}))\bigr] < 1+2d(\sum_{i=1}^{k}\Vol(B_{i}))=1+2d\Vol(J). \qedhere
  \end{align*}
\end{proof}
We have now shown the first two ingredients for our framework: 
the problem is space related and we gave a suitable online algorithm. 
The final piece~--~an offline approximation algorithm~--~is given by the
algorithm of Bansal \emph{et.~al.}~\cite{bansal2016vector}~which is a $ \ln(d+1) +
0{.}807  + \epsilon$-approximation. 
We can thus use \autoref{thm:main} with $\gamma=\ln(d+1) + 0{.}807 +
\epsilon$ and $\beta=2d$ to conclude the following theorem.

\begin{theorem}
  For every $d\geq 1$, there is a robust online algorithm for the dynamic
  $d$-dimensional Vector Packing problem that is $\ln(d+1) + 0{.}807  +
  \epsilon$-competitive and has amortized migration factor $O(1/\epsilon)$.
\end{theorem}



\bibliography{lib}

\end{document}